\documentclass[a4paper]{article}
\usepackage{amsmath}
\usepackage{amssymb}
\usepackage{amsthm}

\usepackage{graphicx}

\newtheorem{thm}{Theorem}
\newtheorem{lem}{Lemma}
\newtheorem*{cor*}{Corollary}

\newcommand{\manifold}{\mathcal{M}}
\newcommand{\lightcone}{\mathcal{H}}

\usepackage{authblk}
\DeclareMathOperator {\tr}{tr}

\usepackage{changes}
\usepackage{amsmath,mathtools}

\setdeletedmarkup{\cancel{#1}}
\usepackage[margins]{trackchanges}
\addeditor{ref1}
\addeditor{ref2}
\addeditor{ref3}
\addeditor{LNS}

\begin{document}
 \title{On almost {Ehlers-Geren-Sachs} theorems}

\author[1]{Ho Lee\footnote{holee@khu.ac.kr}}
\author[2]{Ernesto Nungesser\footnote{em.nungesser@upm.es}}
\author[3]{John Stalker\footnote{stalker@maths.tcd.ie}}
\affil[1]{Department of Mathematics and Research Institute for Basic Science, Kyung Hee University, Seoul, 02447, Republic of Korea}
\affil[2]{M2ASAI, Universidad Polit\'{e}cnica de Madrid, ETSI Navales, Avda. de la Memoria, 4, 28040 Madrid, Spain}
\affil[3]{School of Mathematics, Trinity College, Dublin 2, Ireland}

\maketitle
\begin{abstract}
We show assuming small data that massless solutions to the reflection symmetric Einstein-Vlasov system with Bianchi VII$_0$ symmetry which are not locally rotational symmetric, can be arbitrarily close to  and will remain close to isotropy as regards {to} the shear. However in general the shear will not tend to zero and the Hubble normalised Weyl curvature will blow up. {This generalises} the work \cite{NHW,WHU}, which considered a non-tilted radiation fluid to the massless Vlasov case. This represents another example of the fact that almost {Ehlers-Geren-Sachs} theorems do not hold in general and that collisionless matter behaves differently than a perfect fluid.
 \end{abstract}
 
 \section{Introduction}
 
Since the microwave background is almost isotropic it is natural to consider an isotropic matter distribution. What are the consequences of this assumption for the space-time?  The theorem of Ehlers, Geren and Sachs \cite{EGS} gave an answer to this question, proving that for collisionless matter the space-time has to be either stationary or Robertson-Walker. This was generalised later to the Boltzmann case \cite{TE}. 

What happens if the matter distribution is almost isotropic and the universe is non-stationary? One might think that the space-time has to be almost isotropic. Different results were obtained proving almost {Ehlers-Geren-Sachs} theorems following the research line initiated in \cite{SME}.  

However,  these results were obtained under {the assumption that the dimensionless time and spatial derivatives of the CMB temperature multipoles are bounded by the CMB temperature multipoles themselves, which does not hold in general} \cite{NUWL,R}. This assumption implies in particular that the Hubble normalised Weyl curvature is bounded.

If a cosmological constant is present, non-linear stability and isotropisation of solutions {have} been shown for a variety of matter models {(}cf. \cite{JTV,Ring} and references therein{)}. What happens if no cosmological constant is present?

In \cite{WHU} it was shown that Bianchi VII$_0$ solutions which are not locally rotational symmetric (LRS) with a non-tilted perfect fluid isotropise as regards {to} the shear but do not isotropise as regards {to} the Hubble normalised Weyl curvature. They also showed that self-similarity breaking occurs for any non-LRS Bianchi VII$_0$ solution. This was proven for any non-tilted fluid except for a radiation fluid. The latter case was proven in \cite{NHW}. Afterwards these results were extended to a tilted fluid in \cite{CHe,HervikVII0, LDW}.

One could argue that these results are special since the matter model is a perfect fluid, but more recently the massive case for reflection symmetric solutions to the Einstein-Vlasov system with Bianchi VII$_0$ symmetry was covered in \cite{LN2}. Since the microwave background comes from massless particles it is of special interest to treat this case and in the present paper we obtain a similar result for the massless case assuming small data. {More specifically we assume that the shear, certain variables related to the second and fourth order moments normalised by the energy density and the inverse of some curvature variable are small.}

The massless case (radiation fluid) when treating a non-tilted fluid \cite{NHW}  or a tilted fluid \cite{LDW} for a fluid with a linear equation of state $P=(\gamma-1)\rho$ was more complicated in the case $\gamma=\frac43$. The reason is that in this case some eigenvalues of the linear part vanish and center manifold theory has to be used. As a consequence the shear variables have a polynomial decay for $\gamma =\frac43$ \cite{LDW,NHW} while they have an exponential decay for $\frac23 <\gamma < 2$ with $\gamma \neq \frac43$ \cite{LDW,WHU} .

In the Vlasov case the massive and the massless case are also different. In the massive case the solutions tend to the same behaviour as in the dust case. One has a system of differential equations where the linear part has negative eigenvalues. In the {present} massless case the linear part of the system of differential equations has an eigenvalue which vanishes, but now instead of an equilibrium point there is an equilibrium line. 

Nevertheless in both massive and massless Vlasov case there is an exponential decay. The behaviour at late times thus differs from that of a radiation fluid where a polynomial decay was found \cite{NHW}.

Moreover in \cite{LDW,NHW} it was shown that the shear tends to zero both in the non-tilted and the tilted case. We show here that for collisionless {matter}, although the shear will always remain small, it nevertheless does not tend to zero, which is another difference with respect to a fluid.

Finally the result of this paper extends the known results concerning the late time behaviour of massless solutions to the Einstein-Vlasov system with Bianchi symmetries {for which there have been some recent progress} \cite{B,BFH,LNT2}.

The structure of the paper is as follows. In Section \ref{masslessEV} we present the massless Einstein-Vlasov system with Bianchi symmetry following \cite{Ring}. In particular we develop the evolution equations of two key variables \eqref{wpm} and \eqref{defxi} which are respectively the components of the second and the fourth order moments of the particle distribution function normalised by the energy density. In Section \ref{rswhu} we obtain the equations for the massless Einstein-Vlasov system with Bianchi VII$_0$ and reflection symmetry and introduce some variables adapted to the problem following \cite{WHU}. Section \ref{main} is the core of this paper where Theorem \ref{bootthm} is proven. In particular the late time behaviour of solutions to the massless Einstein-Vlasov system with Bianchi VII$_0$ and reflection symmetry assuming small data is obtained. This is achieved using a bootstrap argument and analysing a reduced system of equations which we have called the truncated system. The result on the truncated system is Lemma 1. The section finishes with a corollary which shows that the Hubble normalised Weyl curvature blows up for large times.  This implies in particular that for massless collisionless matter, `an almost isotropic cosmic microwave temperature does not imply an almost isotropic universe' \cite{NUWL}, since the conditions of small shear and small Hubble normalised Weyl curvature one would assume for an almost isotropic universe (cf. (3) of \cite{NUWL}) do not hold. In the last section we discuss the results obtained, make some conclusions and consider future perspectives.
 
Throughout the paper we assume that Greek letters run from $0$ to $3$, while Latin letters vary from $1$ to $3$, and also follow the sign conventions of \cite{Ring}.
 
 \section{The massless Einstein-Vlasov system}\label{masslessEV}
 
In this section we begin by introducing the massless Einstein-Vlasov system in general. Then, in Section \ref{bianchi} we particularise to the massless Einstein-Vlasov system with Bianchi symmetry, define the variables \eqref{wpm} and obtain some bounds on the latter variables. Finally in Section \ref{vlasovbianchi} we introduce several variables related to higher order moments and their time derivatives.

Consider a four-dimensional oriented and time oriented Lorentzian manifold $(\manifold, {^4g})$ and a particle distribution function $f$. Then,  the massless Einstein-Vlasov system is written as
\begin{align*}
G_{\alpha\beta}&= T_{\alpha \beta},\\
\mathcal{L} f&=0,
\end{align*}
where $G_{\alpha\beta}$ and $T_{\alpha\beta}$ refer to components of the the Einstein tensor and the stress energy tensor respectively and $\mathcal{L}$ the Liouville operator. In the present case we assume that the stress energy tensor is described as follows:
\begin{align*}
T_{\alpha\beta}=  \int_{\lightcone 
\setminus  \{0\} } \chi p_{\alpha} p_{\beta},
\end{align*}
where $p^{\alpha}$ are the four momenta of the particles which will be future oriented according to the time orientation we will introduce in \eqref{fourmetric} and the integration is over the future pointing light-cone $\lightcone$ at a given space-time point which is defined by
\[
p_{\alpha} p_{\beta}  \,  {^ 4 g}^{\alpha\beta}=0, \quad p^0>0
\]
with the apex removed. The quantity $\chi$ is the particle distribution function multiplied by the Lorentz invariant measure.

The basic equations we will use can be found in Sections 7.3--7.4 and Chapter 25 of \cite{Ring}. We also refer to this book for an introduction to the Einstein-Vlasov system. Let $\Sigma$ be a spacelike hypersurface in $\manifold$ with $n$ its future directed unit normal. Let $g$ be the Riemannian metric induced on $\Sigma$ by $^4g$. We define the second fundamental form as $k(X,Y)=g(\nabla_X n, Y)$ for vectors $X$ and $Y$ tangent to $\Sigma$, where $\nabla$ is the Levi-Civita connection associated with $^4g$. The Hamiltonian and momentum constraints are as follows:
\begin{align*}
&R-{k}_{ij} {k}^{ij}+ {k}^2=2 \rho,\\
&\overline{\nabla}^j {k}_{ji}-\overline{\nabla}_i  {k}= -{J}_i,
\end{align*}
where $k=k_{ab}g^{ab}$ is the trace of the second fundamental form $k(X,Y)$ of $\Sigma$, $R$ and $\overline{\nabla}$ are the scalar curvature and the Levi-Civita connection of ${g}$ respectively, and matter terms are given by $\rho= T_{\alpha\beta}n^{\alpha}n^{\beta}$ and $J_i X^i=- T_{\alpha\beta}n^{\alpha}X^{\beta}$ for $X$ tangent to $\Sigma$.

\subsection{The massless Einstein-Vlasov system with Bianchi symmetry}\label{bianchi}

A Bianchi spacetime is defined to be a spatially homogeneous spacetime whose isometry group possesses a three-dimensional subgroup that acts simply transitively on spacelike orbits. A Bianchi spacetime admits a Lie algebra of Killing vector fields. These vector fields are tangent to the group orbits, which are the surfaces of homogeneity. Using a left-invariant frame, the metric induced on the spacelike hypersurfaces depends only on the time variable. Let $G$ be the three-dimensional Lie group, $e_i$ a basis of the Lie algebra, and $\xi^i$ the dual of $e_i$. The metric of the Bianchi spacetime in the left-invariant frame is written as 
\begin{align}\label{fourmetric}
^4 g =-dt \otimes dt + g_{ij}\xi^i \otimes \xi^j
\end{align}
on $\manifold=I \times G$ with $e_0=\frac{\partial}{\partial t}$ future oriented. Define the structure constants by
\begin{align*}
[e_i,e_j]= C_{ij}^l e_l.
\end{align*}
We will need equations (25.17)--(25.18) of \cite{Ring} (without scalar field) with the notation $T_{ab}=S_{ab}$:
\begin{align}
&\dot{g}_{ab}=2k_{ab},\label{a} \\
&\dot{k}_{ab}=-R_{ab}+2 k^i_a k_{bi} -k\, k_{ab}+S_{ab}\label{EE},
\end{align}
where the dot means the derivative with respect to time $t$ and $R_{ab}$ are the components of the Ricci tensor associated to the induced 3-metric. Note that in the massless case 
\begin{align}\label{massless}
g^{ab} S_{ab}=\rho.
\end{align}
Since $k$ does not depend on spatial variables, the constraint equations are as follows: 
\begin{align}
&R-{k}_{ij} {k}^{ij}+ {k}^2=2 \rho,\label{CE1} \\
&\overline{\nabla}^j {k}_{ji}= -{J}_i.\label{CE2}
\end{align}
Moreover, following the conventions of \cite{Ring}, we have
\begin{align*}
\nabla_{e_j} e_l = \Gamma^i_{jl} e_i,
\end{align*}
 and the connection coefficients can be expressed in terms of the structure constants \cite{Ring}:
\begin{align}\label{Gamma}
    \Gamma^{i}_{jl}=\frac12g^{mi}(-C^n_{lm}g_{nj}+C^{n}_{mj}g_{ln}+C^{n}_{jl}g_{nm}).
\end{align}
From the last equation we obtain 
\begin{align}\label{sumcon}
\Gamma^i_{il}=C^n_{nl}, \quad \Gamma^{i}_{jj}= g^{mi} g_{nj} C^n_{mj}, \quad \Gamma^i_{jl} g^{jl}= g^{mi} C_{ml}^l,
\end{align}
since $g^{mi}C_{mi}$ vanishes due to the symmetry of the metric and the antisymmetry of the structure constants and there is no summation on the index $ j $ in the second expression. 
Moreover the only non-zero components of $\Gamma^{\alpha}_{\beta\gamma}$ are $\Gamma^{a}_{bc}$ and
\begin{align}\label{Ring254}
\Gamma^0_{ab}=k_{ab}, \quad \Gamma^b_{a0}=\Gamma^b_{0a}=k^b_a.
\end{align}

Now, we wish to express the momentum constraint in terms of the connection coefficients. We have
\begin{align*}
\nabla^a{k}_{bc}= g^{ad} \nabla_d k_{bc} = g^{ad} \left(\frac{\partial  k_{bc}}{\partial x^d} - \Gamma^f_{db}     k_{fc}  -\Gamma^f_{dc}  k_{bf}\right)= -g^{ad} (\Gamma^f_{db} k_{fc} +\Gamma^f_{dc}  k_{bf}),
\end{align*}
which implies
\begin{align*}
J_i= g^{ad} (\Gamma^f_{da} k_{fi} +\Gamma^f_{di}  k_{af}),
\end{align*}
 and that one can express $J_i$ in terms of the metric, the second fundamental form and the structure constants. Using \eqref{sumcon} and considering the Bianchi A case where $C_{ml}^l=0$, the last equation turns to
 \begin{align*}
 J_i =k^m_n C^n_{mi}.
 \end{align*}
Now denote by $\epsilon_{ijk}$ the standard permutation symbol. For Bianchi A spacetime we have that (cf. E.1 of \cite{Ring} with $a_j=0$) 
\begin{align}\label{symm}
C^{n}_{mi}= \epsilon_{mil}n^{ln},
\end{align}
where $n^{lk}$ is a symmetric matrix also called the structure constant matrix which characterises the Bianchi type. We thus have:
\begin{align}\label{J}
J_i =  \epsilon_{mil}n^{ln}k^m_n.
\end{align}

Below, we collect and derive several useful equations. Using the fact that 
\begin{align}\label{inverse}
\dot{g}^{ab}=-2k^{ab},
\end{align}
we obtain
\begin{align}\label{MV}
\dot{k}^a_b= -R^a_b-k\,k^a_b +S^a_b,
\end{align}
and the trace of (\ref{MV}) with respect to the induced metric, contracting with $ \delta_a^b $ gives us
\begin{align}\label{im}
\dot k=-R-k^2+\rho.
\end{align}
It is convenient to express the second fundamental form as
\begin{align*}
k_{ab}=\sigma_{ab}+H g_{ab},
\end{align*}
where $\sigma_{ab}$ is trace free, and $H=\frac13 k$ is the Hubble parameter. Then \eqref{im} becomes
\begin{align*}
\dot{H}=-3H^2-\frac{1}{3}R+\frac13 \rho,
\end{align*}
and \eqref{CE1} becomes
\begin{align}\label{omega}
\Omega= \frac{\rho}{3H^2}= 1+\frac16\bar{R}-\frac16 F,
\end{align}
where $\bar{R}=\frac{R}{H^2}$ and $F=\frac{\sigma_{ab}\sigma^{ab}}{H^2}$.

In terms of the trace free part \eqref{MV} transforms into
\begin{align*}
\dot{\sigma}^a_b= -3H \sigma^a_b -R^a_b +S^a_b-(3H^2+\dot{H})\delta^a_b= -3H \sigma^a_b -R^a_b +S^a_b-(-\frac{1}{3}R+\frac13 \rho)\delta^a_b,
\end{align*}
or
\begin{align}\label{sigma}
\dot{\sigma}^a_b&=    -3H\sigma^a_b - r^{a}_b+\pi^a_b,
\end{align}
where $r^{a}_b$ and $\pi^a_b$ are the trace free part of $R^a_b$ and $S^a_b$ respectively.

By the constraint equation \eqref{CE1} one can eliminate the energy density such that (\ref{im}) reads:
\begin{align}\label{in}
\dot k=-\frac12 R -\frac12 k^2-\frac12 k_{ij}k^{ij}.
\end{align}
Using the trace free part of $k_{ab}$ and the Hubble variable we obtain
\begin{align}\label{H-1}
\frac{d}{dt}(H^{-1})=-\frac{\dot{H}}{H^2}=2+\frac16 \bar{R} + \frac{\Sigma_a^b\Sigma^a_b}{6},
\end{align}
where we have defined
\begin{align*}
\Sigma_a^b=\frac{\sigma_a^b}{H}.
\end{align*}
It is convenient to introduce a dimensionless time variable $\tau$ as follows:
\begin{align}\label{deftau}
    \frac{dt}{d\tau}=H^{-1},
\end{align}
and denote derivation with respect to that variable by a prime.  Sometimes it is also useful to use the variable $q$:
\begin{align*}
    q=-1-\frac{\dot{H}}{H^2}= 1+\frac16\bar{R}+\frac16 F,
\end{align*}
where we have used \eqref{H-1} in the last equation.
The evolution equation of ${\Sigma}^a_b$ is then
\begin{align*}
({\Sigma}^a_b)'=-\left(3+\frac{\dot{H}}{H^2}\right)\Sigma^a_b+\frac{\pi^a_b-r^a_b}{H^2}=\left(q-2\right)\Sigma^a_b +\frac{\pi^a_b-r^a_b}{H^2}.
\end{align*}
Using \eqref{H-1} and \eqref{omega} we have
\begin{align*}
(\Sigma^a_b)'=\left(-1+\frac16\bar{R}+\frac16 F \right)\Sigma^a_b+3\left(1+\frac16\bar{R}-\frac16F\right)\frac{\pi^a_b-r^a_b}{\rho}.
\end{align*}

Since $\Sigma^a_b$ is trace free sometimes it is convenient to work with $\Sigma_+$ and $\Sigma_-$ as was done in \cite{LN2} which are defined by
\begin{align*}
\Sigma_+=\frac{1}{2H}\left (\sigma^2_2+\sigma^3_3\right), \quad
\Sigma_-=\frac{1}{2\sqrt{3}H}\left (\sigma^2_2-\sigma^3_3\right),
\end{align*}
so that 
\begin{align}\label{Sigma}
(\Sigma^1_1,\Sigma^2_2,\Sigma_3^3)= (-2\Sigma_+,\Sigma_++\sqrt{3}\Sigma_-,\Sigma_+-\sqrt{3}\Sigma_-).
\end{align}

Using $H^2=\frac{\rho}{3\Omega}$ we define $w_\pm$ analogously to $\Sigma_{\pm}$ by
\begin{align}\label{wpm}
    &w_+ = \frac{\pi^2_2+\pi^3_3}{2\rho},\\
   \label{wpm2} &w_- = \frac{\pi^2_2-\pi^3_3}{2\sqrt{3}\rho}.
\end{align}
Note that by definition $\pi^2_2 = S^2_2 - \frac13 \tr S$,  $\pi^3_3 = S^3_3 - \frac13 \tr S$ and using \eqref{massless} we obtain
\begin{align*}
w_+ = \frac{S^2_2-\frac13 \rho+  S^3_3-\frac13 \rho}{2\rho},\quad w_- = \frac{S^2_2-  S^3_3}{2\sqrt{3} \rho}.
\end{align*}
Using the fact that $0 \leq S^2_2+S_3^3 \leq  \tr S= \rho$ and $0\leq S^2_2 \leq \tr S$, $0\leq S^3_3 \leq \tr S$, we obtain the following bounds for $w_{\pm}$:
\begin{align}\label{boundw}
-\frac13 \leq w_+ \leq \frac 16, \quad -\frac{1}{2\sqrt{3}} \leq w_- \leq \frac{1}{2\sqrt{3}}.
\end{align}
The evolution equations of $\Sigma_{\pm}$ are thus
\begin{align}
\label{plus}&{\Sigma}_+'=(q-2)\Sigma_++\frac{2R-3(R^2_2+R^3_3)}{6H^2}+3w_+\Omega,\\
\label{minus}&{\Sigma}_-'=(q-2)\Sigma_-+\frac{R_3^3-R^2_2}{2\sqrt{3}H^2}+3w_-\Omega.
\end{align}

\subsection{Vlasov equation with Bianchi symmetry}\label{vlasovbianchi}

We use a left-invariant and assume that the spacetime has a Bianchi symmetry so that the particle distribution function $f$ can be written as a function of $ t $ and $ p $. Moreover, we assume that the particle distribution function $f$ has compact support for simplicity. Since $g_{00}=g^{00}=-1$ and $g^{0a}=0$, we have $p^0=-p_0=\sqrt{p_ap_bg^{ab}}$, $\rho=T_{00}$, and $J_{a}=-T_{0a}$. The frame components of the energy-momentum tensor are thus
\begin{align*}
&\rho=(\det g)^{-\frac12}  \int f p^0 dp,\quad J_{i}=(\det g)^{-\frac12} \int f p_i  dp,\\
&S_{ij}=(\det g)^{-\frac12} \int f \frac{p_i p_j} {p^0}dp,
\end{align*}
where the particle distribution function is understood as $f=f(t,p)$ with $p=(p_1,p_2,p_3)$ and $dp=dp_1dp_2dp_3$.

Using the expressions for the connection coefficients and the antisymmetry of the structure constants the Vlasov equations as expressed in (25.14) of \cite{Ring} where $f$ depends on $p^i$ turns into
\begin{align*}
p^0 \frac{\partial f}{\partial t} = \left(2 k^i_b p^0 p^b + g^{mi} C^{d}_{ma} p^a p_d \right) \frac{\partial f}{\partial p^i}.
\end{align*}
Considering $f$ as function of $p_i$, i.e. $f(t,p^i)=\bar{f}(t,p_i)=\bar{f}(t,g_{ij}p^j)$, we have that 
\begin{align*}
&\frac{\partial f}{\partial t}= \frac{\partial \bar{f}}{\partial t}+ \frac{\partial \bar{f}}{\partial p_i}\dot{g}_{ij}p^j, \\
&\frac{\partial f}{\partial p^j}= \frac{\partial \bar{f}}{\partial p_i}g_{ij},
\end{align*}
from which it follows that the equation for the particle distribution function in terms of $p_i$ dropping the bar by a slight abuse of notation is
\begin{align}\label{vlasovstructure}
    p^0\frac{\partial f}{\partial t}+C^d_{ba}p^bp_d \frac{\partial f}{\partial p_a}=0.
\end{align}

When treating the Bianchi I case in \cite{LNT2} the following quantity was considered:
\begin{align}\label{w}
w^j_i = \frac{S^j_i}{\rho}=\frac{\int f p_i p_a g^{aj} (p^0)^{-1}dp}{\int f p^0 dp}.
\end{align}
In the Bianchi I case $f$ does not depend on $t$, but in general it does. Let us consider the terms where the time derivative of $f$ is involved. Consider first the following quantity:
\begin{align*}
D=\int \frac{\partial f}{\partial t} p^0 dp.
\end{align*}
Using the Vlasov equation \eqref{vlasovstructure} and integrating by parts we have
\begin{align*}
D &=\int \frac{\partial f}{\partial t} p^0 dp=-  C^d_{ba}\int p^bp_d \frac{\partial f}{\partial p_a} dp =   C^d_{ba}\int f \frac{\partial }{\partial p_a} (p^bp_d) dp\\
&= C^d_{ba}\int f (\delta^a_d p^b  + p_d g^{ab})dp = C^a_{ba}\int f  p^b dp,
\end{align*}
where we have used again the fact that $g^{ab}$ is symmetric while $C^d_{ba}$ is antisymmetric which implies $g^{ab}C^d_{ba}=0$. Moreover for all Bianchi A spacetimes $C^a_{ba}=0$ which means that $D$ vanishes in that case.

Consider now 
\begin{align*}
V^j_i  = \int \frac{\partial f}{\partial t} p_i p_f g^{fj} (p^0)^{-1} dp.
\end{align*}
Using again the Vlasov equation and integrating by parts
\begin{align*}
V^j_i & = \int \frac{\partial f}{\partial t} p_i p_f g^{fj} (p^0)^{-1} dp =-C^d_{ba} \int \frac{\partial f}{\partial p_a} \frac{ p^bp_d p_i p_f g^{fj}}{ (p^0)^{2}} dp \\ 
& = C^d_{ba}g^{fj}g^{eb} \int f \frac{\partial }{\partial p_a} \left[ \frac{ p_e p_d p_i p_f }{ (p^0)^{2}} \right] dp.
\end{align*}
If $e$ is equal to $a$ this term vanishes due to the antisymmetry of the structure constants and for $d=a$ the term vanishes for Bianchi A spacetimes. Let us consider Bianchi A spacetimes. In that case
\begin{align*}
V^j_i= C^d_{ba} g^{fj}g^{eb}\int f p_e p_d \frac{\partial }{\partial p_a} \left[ \frac{  p_i p_f }{ (p^0)^{2}} \right] dp = C^d_{ba} g^{fj}\int f p^b p_d \frac{\partial }{\partial p_a} \left[ \frac{  p_i p_f }{ (p^0)^{2}} \right] dp. 
\end{align*}
Now
\begin{align*}
\frac{\partial (p^0)^2}{ \partial p_a} = 2 p^a.
\end{align*}
The resulting expression in the previous integral with the term $C^d_{ba}p^bp^a$ vanishes due to the antisymmetry of the structure constants. Thus
\begin{align}\label{V}
V^j_i=  C^d_{ba} g^{fj}  \int f \frac{p^b p_d}{(p^0)^{2}} \frac{\partial }{\partial p_a} \left[  p_i p_f \right] dp = C^d_{ba}g^{fj}  \int f \frac{p^b p_d}{(p^0)^{2}} \left[ \delta^a_i p_f +\delta^a_f p_i\right] dp.
\end{align}
Define
\begin{align}\label{W}
W^j_i= \frac{V^j_i}{H\int f p^0 dp}.
\end{align}
The derivative of $w_i^j$ with respect to $\tau$ is then:
\begin{align}\label{wev}
(w_i^j)'=-2w_i^a \Sigma^j_a+w^b_a \Sigma^a_b w_i^j+ \Sigma^{c}_d\xi_{ic}^{jd}+W^j_i,
\end{align}
where 
\begin{align}\label{defxi}
\xi_{ic}^{jd}=\frac{\int f p_ip^jp_cp^d (p^0)^{-3}dp}{\int fp^0 dp}.
\end{align}
Moreover, the derivative of ${\xi}_{ic}^{jd}$ is
\begin{align}\label{evxi}
\left({\xi}_{ic}^{jd}\right)'=-2\Sigma^d_f {\xi}_{ic}^{jf}-2\Sigma^j_f{\xi}_{ic}^{fd}+3\Sigma^a_b \Phi_{ica}^{jdb}+ {\xi}_{ic}^{jd} w_a^b \Sigma^a_b + Z_{ic}^{jd},
\end{align}
where
\begin{align*}
\Phi_{ica}^{jdb} =  \frac{\int {f} p_ip^jp_cp^dp_ap^b (p^0)^{-5} dp}{\int {f} p^0 dp},
\end{align*}
and
\begin{align}\label{ZZ}
Z_{ic}^{jd}= \frac{\int \partial_ t f p_ip^jp_cp^d (p^0)^{-3}dp}{H\int fp^0 dp}.
\end{align}

Consider the integral in the numerator of \eqref{ZZ} and call it $X_{ic}^{jd}$. Doing a similar procedure as for $V^j_i$, i.e. using the Vlasov equation and integrating by parts yield:
\begin{align}
X_{ic}^{jd}&= \int \partial_ t f p_ip^jp_cp^d (p^0)^{-3}dp = - C^g_{ba}  \int  \frac{\partial f}{\partial p_a} p^bp_g p_ip^jp_cp^d (p^0)^{-4}dp \nonumber \\
                & =C^g_{ba} \int f \frac{\partial }{\partial p_a} \left[\frac{p^bp_g p_ip^jp_cp^d }{(p^0)^{4}}\right] dp = C^g_{ba} \int f \frac{p^b p_g}{(p^0)^4}\frac{\partial }{\partial p_a}  (p_ip^jp_cp^d) dp, \label{X}
\end{align}
where the last equality is obtained by the same considerations as when treating $V^j_i$.

In the following we will consider that we are close to the case that $f$ has some symmetries. In that case $\xi_{ic}^{jd}$ takes some specific values which we will denote by $\hat{\xi}_{ic}^{jd}$. 
Let us consider the trace free part $\tilde{w}^i_j=w^i_j-\frac13\delta_j^i$ and define 
\begin{align}\label{xi}
\tilde{\xi}_{ic}^{jd}=\xi_{ic}^{jd}-\hat{\xi}_{ic}^{jd}. 
\end{align}
Then  \eqref{wev}  turns into
\begin{align}
\label{evw}(\tilde{w}_i^j)'=-\frac23\Sigma^j_i-2\tilde{w}_i^a\Sigma^j_a+\tilde{w}^b_a \Sigma^a_b (\tilde{w}_i^j+\frac13\delta^j_i)+ \Sigma^{c}_d \left(\tilde{\xi}_{ic}^{jd}+\hat{\xi}_{ic}^{jd}\right)+W^j_i.
\end{align}
In particular for the terms $w_+=\frac12(\tilde{w}^2_2+\tilde{w}^3_3)$, $w_-=\frac{1}{2\sqrt{3}}(\tilde{w}^2_2-\tilde{w}^3_3) $ we have
\begin{align}
\nonumber w_+' =&-\frac23 \Sigma_+- \tilde{w}_2^a\Sigma^2_a - \tilde{w}_3^a\Sigma^3_a +\tilde{w}^b_a \Sigma^a_b \left(w_++\frac13\right)\\
\label{w+1}&+\frac12 \Sigma^{c}_d \left(\tilde{\xi}_{2c}^{2d}+\hat{\xi}_{2c}^{2d}+\tilde{\xi}_{3c}^{3d}+\hat{\xi}_{3c}^{3d}\right)+W_+,\\
\nonumber w_-' =&-\frac23 \Sigma_-- \frac{1}{\sqrt{3}}(\tilde{w}_2^a\Sigma^2_a - \tilde{w}_3^a\Sigma^3_a) +\tilde{w}^b_a \Sigma^a_b w_- \\
\label{w-1}&+\frac{1}{2\sqrt{3}} \Sigma^{c}_d \left(\tilde{\xi}_{2c}^{2d}+\hat{\xi}_{2c}^{2d}-\tilde{\xi}_{3c}^{3d}-\hat{\xi}_{3c}^{3d}\right)+W_-,
\end{align}
where we have used the notation
\begin{align*}
&W_+= \frac12(W^2_2+W^3_3),\\
&W_-=\frac{1}{2\sqrt{3}}(W^2_2-W^3_3).
\end{align*}

\section{The equations for Bianchi VII$_0$ with reflection symmetry for massless Vlasov particles close to the isotropic case}\label{rswhu}

In the following we will deduce the relevant equations for Bianchi VII$_0$ with reflection symmetry. We will exclude the LRS case, since that case reduces to Bianchi I, which presents a completely different behaviour and was already studied in \cite{LNT2}.

We will assume that we are close to the isotropic case. This means that as in \cite{LNT2} the only non-vanishing expressions for $\hat{\xi}^{ab}_{cd}$ are
(suspending the Einstein summation convention for the next expressions):
\begin{align*}
&\hat{\xi}_{aa}^{aa}=\frac{1}{5}, \quad a=1,2,3\\ &\hat{\xi}_{ab}^{ab}=\hat{\xi}_{ab}^{ba}=\hat{\xi}_{aa}^{bb}=\frac{1}{15}, \quad a\neq b, \quad a,b=1,2,3.
\end{align*}
Using these expressions and \eqref{Sigma} we obtain that:
\begin{align*}
\Sigma^c_d \hat{\xi}^{2d}_{2c}& =  \frac{2}{15}\Sigma_+ +\frac{2\sqrt{3}}{15} \Sigma_-,\\
\Sigma^c_d \hat{\xi}^{3d}_{3c}& =  \frac{2}{15}\Sigma_+ -\frac{2\sqrt{3}}{15} \Sigma_-.
\end{align*}
Setting this into the equations for $w_\pm$ \eqref{w+1}--\eqref{w-1} we have
\begin{align}
\label{w+re}&w_+' =-\frac{8}{15} \Sigma_+- \tilde{w}_2^a\Sigma^2_a - \tilde{w}_3^a\Sigma^3_a +\tilde{w}^b_a \Sigma^a_b \left(w_++\frac13\right)+\frac12 \Sigma^{c}_d \left(\tilde{\xi}_{2c}^{2d}+\tilde{\xi}_{3c}^{3d}\right)+W_+,\\
&w_-' =-\frac{8}{15} \Sigma_-- \frac{1}{\sqrt{3}}(\tilde{w}_2^a\Sigma^2_a - \tilde{w}_3^a\Sigma^3_a) +\tilde{w}^b_a \Sigma^a_b w_- +\frac{1}{2\sqrt{3}} \Sigma^{c}_d \left(\tilde{\xi}_{2c}^{2d}-\tilde{\xi}_{3c}^{3d}\right)+W_-.\label{w-re}
\end{align}

If $f$ is reflection symmetric (cf. \cite{Rendall}), i.e.  
\begin{align*}
 f(p_1,p_2,p_3,t)=f(-p_1,-p_2,p_3,t)=f(p_1,-p_2,-p_3,t),
\end{align*}
the metric and the second fundamental form are diagonal, the evolution will preserve this symmetry. Stress energy tensor $T_{ij}$, metric $g_{ij}$ and second fundamental form $k_{ij}$ will be diagonal and also $J_a=T_{0a}=0$ since it is an integral over an odd number of momenta which vanishes due to the reflection symmetry.  Since $J_a=0$ we have from \eqref{J} that $k^i_j$ and $n^{ij}$ commute, which means there exists a basis where one can simultaneously diagonalise both $k^i_j$ and $n^{ij}$. If we thus choose eigenvectors of $k^i_j(t_0)$ as the frame and choose $g(t_0)$ to be diagonal since $J_a=0$, everything will remain diagonal. Now if $n^{ij}$ is diagonal we have that the right hand side of \eqref{J} is always zero . This comes from the fact that due to \eqref{symm}, the structure constants cannot have two equal indices in that case. As a consequence the momentum constraint will be trivially satisfied for all time.

Moreover due to this symmetry the quantity $W^i_j$ defined in \eqref{W} and the quantity $Z_{ic}^{jd}$ defined in  \eqref{ZZ} vanish as well. The reason is the same as for $J_a$. The integrals involved in the numerator for $W^i_j$ and $Z_{ic}^{jd}$ are $V^i_j$ and $X_{ic}^{jd}$ which are integrals over an odd number of momenta which due to the reflection symmetry vanish.

 As a consequence, for the reflection symmetric case, the equations for $\tilde{w}^j_i$ and ${\xi}_{ic}^{jd}$ will be identical to the equations for the Bianchi I case treated in \cite{LNT2}. In particular the equations for $w_{\pm}$ \eqref{w+re}--\eqref{w-re} in the reflection symmetric case using
 \begin{align*}
 &\tilde{w}^2_2\Sigma^2_2+\tilde{w}^3_3\Sigma^3_3= 2w_+\Sigma_++6w_-\Sigma_-,\\
 & \tilde{w}^2_2\Sigma^2_2-\tilde{w}^3_3\Sigma^3_3=2\sqrt{3}\Sigma_-w_++2\sqrt{3}\Sigma_+w_-,\\
 &\tilde{w}^b_a \Sigma^a_b=  6w_+\Sigma_++6w_-\Sigma_-,
 \end{align*}
and the fact that $W_+=W_-=0$ are 
\begin{align}
\label{w++}&{w}_+'=-\frac{8}{15}\Sigma_+   - 4w_-\Sigma_-  +6(w_+\Sigma_++w_-\Sigma_-) w_++R_+, \\
\label{w--}&{w}_-'=-\frac{8}{15} \Sigma_- - 2(w_+\Sigma_- + w_- \Sigma_+) +6(w_+\Sigma_++w_-\Sigma_-) w_-+R_-,
\end{align}
with
\begin{align*}
R_+&=\frac12\left[\Sigma_+\left(2\tilde{\xi}_{32}^{32}-2\tilde{\xi}_{21}^{21}-2\tilde{\xi}_{31}^{31}+ \tilde{\xi}_{22}^{22}+\tilde{\xi}_{33}^{33}\right)+\sqrt{3}\Sigma_- \left(\tilde{\xi}_{22}^{22}-\tilde{\xi}_{33}^{33}\right)\right],\\
R_-&= \frac{1}{2\sqrt{3}}\left[\Sigma_+\left(-2\tilde{\xi}_{21}^{21}+2\tilde{\xi}_{31}^{31}+\tilde{\xi}_{22}^{22}-\tilde{\xi}_{33}^{33}\right)+\sqrt{3}\Sigma_-  \left(\tilde{\xi}_{22}^{22}-2\tilde{\xi}_{32}^{32}+\tilde{\xi}_{33}^{33}\right)\right].
\end{align*}

The constraint equation is
\begin{align*}
    \Omega=\frac{\rho}{3H^2}= 1-\Sigma_+^2-\Sigma_-^2 + \frac{R}{6H^2}.
\end{align*}
For the reflection symmetric massless Bianchi VII$_0$ case we have the same equations for the curvature variables as in the massive case \cite{LN2}.    Let $(ijk)$ denote a cyclic permutation of $(123)$ and let us suspend the Einstein summation convention for the next three formulas. Introduce $\nu_i$ as the signs depending on the Bianchi type (Table 1 of \cite{CH}). Now define
\begin{align*}
n_i=\nu_i \sqrt{\frac{g_{ii}}{g_{jj}g_{kk}}},
\end{align*}
where the indices $(ijk)$ used in the previous formula are a cyclic permutation of $(123)$, e.g. $n_1= \nu_1 \sqrt{\frac{g_{11}}{g_{22}g_{33}}}$. Using this notation the Ricci tensor is given by (cf. (11a) of \cite{CH}) 
\begin{align*}
R^i_i=\frac{1}{2}[n_i^2 -(n_j-n_k)^2].
\end{align*}
We introduce now the curvature variables as follows:
\begin{align*}
 N_{ii}=\frac{n_i}{H}.
\end{align*}
In the Bianchi VII$_0$ case we have $\nu_1=0$ and $\nu_2=\nu_3=1$, which means that $N_{22}$ and $N_{33}$ are positive definite and the only non-vanishing structure constants are (cf. Appendix E, p. 695 of \cite{Ring}):
\begin{align}\label{sc7}
C^2_{31}=1=-C^2_{13},\quad C^3_{12}=1=-C^3_{21}.
\end{align}
The curvature expressions are
\begin{align*}
& R_{1}^1=R=-\frac{1}{2} (n_2-n_3)^2,\\
&R^2_2=-R^3_3=\frac{1}{2}(n_2^2-n_3^2).
\end{align*}
The relevant equations are
\begin{align*}
N_{22}' = N_{22} (q+2\Sigma_++2\sqrt{3}\Sigma_-),\\
N_{33}' = N_{33} (q+2\Sigma_+-2\sqrt{3}\Sigma_-),
\end{align*}
where $N_{22}>0$ and $N_{33}>0$,
which we transform into the following variables:
\begin{align*}
&N_+=\frac{N_{22}+N_{33}}{2}>0,\\
&N_-=\frac{N_{22}-N_{33}}{2\sqrt{3}},
\end{align*}
which implies
\begin{align*}
\frac{R}{6H^2}=-N_-^2,
\end{align*}
and
\begin{align*}
N_+^2-3N_-^2>0.
\end{align*}
As a consequence
\begin{align*}
\Omega = 1-\Sigma_+^2-\Sigma_-^2 - N_-^2.
\end{align*}
The relevant evolution equations are thus
\begin{align}
&\label{Sigma+}\Sigma_+'=(q-2)\Sigma_+-2N_-^2+3{w}_+\Omega,\\
&\Sigma_-'=(q-2)\Sigma_--2N_+N_-+ 3{w}_-\Omega,\\
&N_+'=(q+2\Sigma_+)N_++6\Sigma_-N_-,\\
&\label{N-}N_-'=(q+2\Sigma_+)N_-+2\Sigma_-N_+,\\
&\Omega'=2\Omega\left(q-1-3\Sigma_+w_+-3\Sigma_-w_-\right).
\end{align}

\subsection{WHU variables}

What characterises the future behaviour of solutions to the Einstein equations coupled to the equations of a perfect fluid \cite{CHe,HervikVII0,LDW,NHW, WHU} or coupled to the equations of collisionless matter with massive particles \cite{LN2} assuming Bianchi VII$_0$ symmetry is the self-similarity breaking at late times. One of the curvature variables, in the present paper denoted as $N_+$, blows up. It has been very convenient to consider the inverse of this quantity to have bounded variables at late times and to introduce another variable $\psi$ to model the oscillations which appear for late times. We apply these transformation to the present case and call them WHU variables, since it was introduced in the paper of Wainwright, Hancock and Uggla \cite{WHU} for the first time. The WHU variables are introduced as follows:
\begin{align*}
&M=\frac{1}{N_+}>0,\\
&N_-=X \sin \psi,\quad X>0,\\
&\Sigma_-=X \cos \psi, \quad X>0,
\end{align*}
so that
\begin{align}\label{Omega}
    \Omega=1-\Sigma_+^2-X^2.
\end{align}
Since $\Omega=\frac{\rho}{3H^2}\geq 0$, we have that
\begin{align}\label{bound1}
\Sigma_+^2+X^2 \leq 1.
\end{align}
The relevant equations for $\Sigma_+$, $M$, $X$ and $\psi$ are obtained from the equations (23)--(26) of \cite{LN2} by making the following replacements in terms of the relevant quantities we use here:
\begin{align*}
&S_+ = \frac{1}{6H^2} (3S^2_2 + 3 S^3_3 - 2S) = \frac{1}{6H^2} (3\pi^2_2 + 3\pi^3_3)= 3w_+ \Omega, \\
&S_- = \frac{1}{2\sqrt{3}H^2} (S^2_2 - S^3_3) = \frac{1}{2\sqrt{3}H^2} (\pi^2_2 - \pi^3_3) = 3w_- \Omega, \\
&q = 1 + \frac{R}{6H^2}+\frac16 F = 1- N_-^2+ \Sigma_+^2+\Sigma_-^2= 1+ \Sigma_+^2 + X^2 \cos 2\psi,\\
&Q =1+ \Sigma_+^2 .
\end{align*}
The relevant equations in the present case are thus
\begin{align}
&\label{M} {M}'=-M[1 + \Sigma_+^2+2\Sigma_++X^2(\cos 2\psi+3M\sin 2\psi)],\\
&\label{+} \Sigma_+'=-X^2-\Sigma_+(1 -\Sigma_+^2)+(1+\Sigma_+)X^2\cos 2\psi +3{w}_+\Omega,\\
&\label{Xev} X'=[\Sigma_+(1+\Sigma_+) +(X^2-1-\Sigma_+) \cos 2\psi ]X+ 3{w}_-\Omega\cos \psi ,\\
&\label{psi} \psi'=2M^{-1}+(1+\Sigma_+)\sin 2\psi - X^{-1}3{w}_- \Omega\sin \psi,\\
&\label{Omegaev}\Omega' = 2( \Sigma_+^2+X^2\cos 2\psi-3 \Sigma_+ w_+ - 3 X w_- \cos \psi)\Omega,\\
&\label{w+}{w}_+'=\left(-\frac{8}{15}+\alpha+6w_+^2\right)\Sigma_+ + \left[w_- ( 6w_+- 4)+\beta\right]X\cos \psi ,\\
&\label{w-}{w}_-'= \left[w_-(-2+6w_+)+\gamma\right]\Sigma_+  +\left(6w_-^2-\frac{8}{15}-2w_+ + \delta \right) X \cos \psi,
\end{align}
where
\begin{align}
\label{alpha}\alpha & = \frac12 \left(2\tilde{\xi}_{32}^{32}-2\tilde{\xi}_{21}^{21}-2\tilde{\xi}_{31}^{31}+ \tilde{\xi}_{22}^{22}+\tilde{\xi}_{33}^{33}\right),\\
 \beta  &= \frac12 \sqrt{3} \left(\tilde{\xi}_{22}^{22}-\tilde{\xi}_{33}^{33}\right),\\
\label{gamma}\gamma & =  \frac{1}{2\sqrt{3}}\left(-2\tilde{\xi}_{21}^{21}+2\tilde{\xi}_{31}^{31}+\tilde{\xi}_{22}^{22}-\tilde{\xi}_{33}^{33}\right), \\
\label{delta} \delta & = \frac12 \left(\tilde{\xi}_{22}^{22}-2\tilde{\xi}_{32}^{32}+\tilde{\xi}_{33}^{33}\right).
\end{align}

We have put the equations \eqref{M}--\eqref{Omegaev} in such a way that they are easy to compare with the equations in \cite{NHW, WHU}. They are identical to the equations in \cite{WHU} if one sets in our equations ${w}_+={w}_-=0$.

The equations for ${\xi}_{ic}^{jd}$, using the fact that $\tilde{w}^i_j$ and $\Sigma^i_j$ are tracefree and that $Z^{jd}_{ic}$ defined in \eqref{ZZ} vanishes in the reflection symmetric case,  simplify using \eqref{evxi} to:
\begin{align}\label{xixi}
\left({\xi}_{ic}^{jd}\right)'=-2\Sigma^d_f {\xi}_{ic}^{jf}-2\Sigma^j_f{\xi}_{ic}^{fd}+3\Sigma^a_b \Phi_{ica}^{jdb} + {\xi}_{ic}^{jd} \tilde{w}_a^b \Sigma^a_b .
\end{align}
More specifically we obtain:
\begin{align}
\label{xi1}&\left({\tilde{\xi}}_{22}^{22}\right)'= \left[(-4+6w_+)\Sigma_++(6w_--4\sqrt{3}) X \cos \psi \right]\left( {\tilde{\xi}}_{22}^{22} + \frac15\right) +3\Sigma^a_b  \Phi_{22a}^{22b},\\
\label{xi2}&\left({\tilde{\xi}}_{33}^{33}\right)'= \left[(-4+6w_+)\Sigma_++(6w_-+4\sqrt{3}) X \cos \psi\right]\left({\tilde{\xi}}_{33}^{33}+\frac15\right)+3\Sigma^a_b \Phi_{33a}^{33b},\\
\label{xi3}&\left(\tilde{{\xi}}_{21}^{21}\right)'=\left[(2+6w_+)\Sigma_++(6w_--2\sqrt{3})X \cos \psi \right]\left(\tilde{\xi}_{21}^{21}+\frac{1}{15}\right)+3\Sigma^a_b \Phi_{21a}^{21b},\\
\label{xi4}&\left(\tilde{{\xi}}_{31}^{31}\right)'=\left[(2+6w_+)\Sigma_++(6w_-+2\sqrt{3})X \cos \psi \right]\left(\tilde{\xi}_{31}^{31}+\frac{1}{15}\right)+3\Sigma^a_b \Phi_{31a}^{31b},\\
\label{xi5}&\left(\tilde{{\xi}}_{32}^{32}\right)'=\left[(-4+6w_+)\Sigma_++6w_-X \cos \psi \right]\left(\tilde{\xi}_{32}^{32}+\frac{1}{15}\right)+3\Sigma^a_b \Phi_{32a}^{32b}.
\end{align}
Note that the system is not closed, since one could consider the evolution equations of $\Phi_{ica}^{jdb}$, which would include other higher order terms. However the considered evolution equations will be sufficient to obtain the desired future asymptotic behaviour.

\section{Future asymptotic behaviour}\label{main}

In this section we prove the main results. We start by stating and proving Theorem \ref{bootthm}. The idea is to assume certain small data. In some sense one can see as the reference solution the radiation fluid solution of Bianchi VII$_0$ spacetimes which are not LRS. We will assume that we are close to the situation as described in Theorem 2.1, 2.3 of \cite{WHU}, Theorem 1 of \cite{NHW} or Theorem 1 of \cite{LDW} as regards the shear and curvature variables. We will assume that $\Sigma_+$, $X$ and $M$ are small. On the other hand for a perfect fluid the relation between pressure and energy density is fixed, which does not hold in the present case. Here we assume that we are close to the values corresponding to a radiation fluid as regards to the second and the fourth order moments normalised by the energy density. All these quantities are zero for a radiation fluid, here we assume they are small. 

The proof of Theorem \ref{bootthm} will be completed using a bootstrap argument and will be finished in Section \ref{closing}. In the proof of the main theorem we will start in Section \ref{estimateM} by obtaining an estimate for $M$ which is a key variable because it is responsible for the self-similarity breaking. Afterwards in Section \ref{suppress} we redefine the variables to suppress the oscillations. In Section \ref{truncated} we obtain Lemma \ref{lemtru} concerning a truncated system which will enable us to finish the proof of  Theorem \ref{bootthm} in Section \ref{closing}. Finally in Section \ref{weyl} we obtain an estimate of the Weyl curvature and certain results on the asymptotic values of the main variables.

In order to obtain the desired result it will be necessary to assume closeness to an equilibrium line given by the curve
\begin{align*}
  -\frac13{X}^2+{w}_+(1-{X}^2)=0.
\end{align*}
This will become clear when treating the truncated system in Section \ref{truncated}. This system is obtained by neglecting the oscillatory terms, the higher order momentum terms $\tilde{\xi}_{ij}^{kl}$ and focusing on the equations for $\Sigma_+$, $X$ and $w_+$. \\

Before establishing the main theorem, we summarise the different variables and the changes of variables introduced. We consider the massless Einstein-Vlasov system with reflection and Bianchi VII$_0$ symmetry, which is not LRS. The initial data are thus the diagonal elements of the metric, the diagonal elements of the second fundamental form and the particle distribution function, which satisfy the constraint equations \eqref{CE1}--\eqref{CE2}. In fact, we have shown that for reflection symmetric data equation \eqref{CE2} is trivially satisfied. The evolution equations of the metric, the second fundamental and the particle distribution function are given by \eqref{a}--\eqref{EE} and  \eqref{vlasovstructure} with \eqref{sc7}. The equations \eqref{a}--\eqref{EE} are first transformed into \eqref{Sigma+}--\eqref{N-} introducing the variables $\Sigma_+$, $\Sigma_-$, $N_+$, $N_-$, $w_+$, $w_-$ and a new dimensionless time variable $\tau$ via \eqref{deftau}. The Vlasov equation \eqref{vlasovstructure} is used to obtain the evolution equations \eqref{w++}--\eqref{w--} of $w_+$ and $w_-$. Then, the variables $\Sigma_-$, $N_+$ and $N_-$ are substituted by the WHU variables $M$, $X$ and $\psi$, which gives the evolution equations \eqref{M}--\eqref{psi} and \eqref{w+}--\eqref{w-}. Finally with the Vlasov equation and the new variables, the evolution equations \eqref{xi1}--\eqref{xi5} of $\tilde{{\xi}}_{22}^{22}$, $\tilde{{\xi}}_{33}^{33}$, $\tilde{{\xi}}_{21}^{21}$, $\tilde{{\xi}}_{31}^{31}$ and $\tilde{{\xi}}_{32}^{32}$ are obtained. Thus, initial data for the metric, the second fundamental form and the particle distribution function give rise to initial data for $M$, $\Sigma_+$, $X$, $\psi$, $w_+$, $w_-$, $\tilde{{\xi}}_{22}^{22}$, $\tilde{{\xi}}_{33}^{33}$, $\tilde{{\xi}}_{21}^{21}$, $\tilde{{\xi}}_{31}^{31}$ and $\tilde{{\xi}}_{32}^{32}$. We also need to define the following quantity:
\begin{align}
 \label{Y} Y =   -\frac13{X}^2+{w}_+(1-{X}^2).
\end{align}
Then, the following theorem holds:

\begin{thm}\label{bootthm}
Consider any $C^\infty$-solution of the massless Einstein-Vlasov system with reflection and Bianchi VII$_0$ symmetry, which is not LRS, given by the equations \eqref{vlasovstructure}, \eqref{sc7}, \eqref{M}--\eqref{psi}, \eqref{w+}--\eqref{w-} and \eqref{xi1}--\eqref{xi5} with initial data satisfying the constraint \eqref{CE1} and the conditions $X(\tau_0)\neq 0$ and $w_-(\tau_0)\neq 0$. There exists a small $ \varepsilon > 0 $ such that if initial data are given by
\[
M(\tau_0), \vert \Sigma_+(\tau_0)\vert, X(\tau_0), \vert w_+(\tau_0)\vert, \vert \tilde{\xi}_{22}^{22}(\tau_0) \vert,  \vert \tilde{\xi}_{33}^{33}(\tau_0)\vert,  \vert \tilde{\xi}_{21}^{21}(\tau_0)\vert, \vert \tilde{\xi}_{31}^{31}(\tau_0) \vert, \vert \tilde{\xi}_{32}^{32}(\tau_0) \vert < \varepsilon,
\]
then the following estimates hold:
\begin{align}
M &= O(\varepsilon e^{(-1+\varepsilon) \tau}), \label{estM}\\
 Y   &= O( \varepsilon e^{(-\frac12 + \varepsilon) \tau}),\\
 \Sigma_+  &=  O( \varepsilon e^{(-\frac12 + \varepsilon) \tau}),\label{estsigma}\\
\ \frac{w_-}{X}  &= O(1). 
\end{align}
\end{thm}

\noindent\emph{Remark}: 
In fact, the specific values of the structure constants for Bianchi $VII_0$ will not be necessary when using the Vlasov equation. It will be sufficient to use the properties of Bianchi A space-times.

\subsection*{Proof of the theorem}

We begin by observing that the relevant quantities are bounded. The quantities $\Sigma_+$, $X$ and $\Omega$ are bounded due to \eqref{Omega}--\eqref{bound1}, $w_+$ and $w_-$ are bounded due to \eqref{boundw}. Moreover $\xi^{ij}_{ij}$ can be bounded by $w^i_i$ or $w^j_j$ (no summation on repeated indices in the previous three quantities) which are bounded by $1$. The quantities $\tilde{\xi}^{ij}_{ij}$ defined in \eqref{xi}  just differ by a constant from $\xi^{ij}_{ij}$. As a consequence $\alpha$, $\beta$, $\gamma$ and $\delta$ defined in \eqref{alpha}--\eqref{delta} are bounded as well. Similarly $\Phi_{ijk}^{ijk}$ is bounded which implies that the derivatives of $\Sigma_+$, $X$, $\Omega$, $w_+$, $w_-$ and $\tilde{\xi}^{ij}_{ij}$ are also bounded. Now, the proof of the main theorem will be given in the following sections.

\subsection{Estimate of $M$}\label{estimateM}
Since $X$ and $w_+$ are small, we can assume that $Y$ is small. We will use a bootstrap argument.  Let us assume that there exists an interval $[\tau_0,\tau_1)$ where the following estimates hold:

\begin{align}
\label{bootm} M(\tau) &\leq  \varepsilon_M,\\
\label{bootsigma} \vert \Sigma_+ (\tau) \vert & \leq \varepsilon_{\Sigma} e^{(-\frac38+ \varepsilon)\tau},\\
\label{bootY} \vert Y(\tau) \vert &\leq  \varepsilon_{Y} e^{(-\frac38+ \varepsilon)\tau},\\
\label{bootw-X}\left \vert \frac{w_-}{X}(\tau) \right \vert & \leq C e^{\frac14 \tau},
\end{align}
where  $\varepsilon_M$, $\varepsilon_{\Sigma}$, $\varepsilon_Y$ are some small quantities all smaller than $\varepsilon$. The last assumption \eqref{bootw-X} will hold initially since $X$ is assumed to be different from zero and $C$ is some arbitrary constant, not necessarily small.

We want to remove the first oscillatory term of the evolution equation of $M$ and therefore we define:
\begin{align} \label{barm}
\bar{M}=\frac{M}{1-\frac14MX^2 \sin 2\psi}.
\end{align}
Then
\begin{align}\label{barM}
\frac{\bar{M}'}{\bar{M}}=\frac{-(1 + \Sigma_+)^2+ MXA}{1-\frac14MX^2 \sin 2\psi},
\end{align}
where 
\begin{align*}
A= \left(-3X +\frac12X'+\frac12X\cos2\psi (1+\Sigma_+)\right) \sin 2\psi-\frac32 \Omega w_- \cos 2\psi \sin \psi
\end{align*}
is a bounded quantity since $X$, $X'$, $\Sigma_+$, $w_-$ and $\Omega$ are bounded.
Using the bounds of the bootstrap assumptions \eqref{bootm}--\eqref{bootsigma} we obtain
\begin{align*}
-1 - C \varepsilon \leq \frac{\bar{M}'}{\bar{M}}  \leq -1 + C \varepsilon,
\end{align*}
which implies that
\begin{align}\label{boundM}
\bar{M}(\tau_0) e^{(-1 -C \varepsilon)(\tau-\tau_0)} \leq \bar{M}(\tau) \leq \bar{M}(\tau_0) e^{(-1 + C \varepsilon)(\tau-\tau_0)},
\end{align}
and
\begin{align*}
M (\tau) &\leq (1-\frac14MX^2 \sin 2\psi)(\tau) \frac{M(\tau_0)}{  (1-\frac14MX^2 \sin 2\psi)(\tau_0)} e^{(-1 + C \varepsilon)(\tau-\tau_0)}\\
&\leq M(\tau_0) \frac{1+C  \varepsilon_M}{1-CM(\tau_0)} e^{(-1 + C \varepsilon)(\tau-\tau_0)}.
\end{align*}
Choosing $M(\tau_0)$ smaller than $\frac12 \varepsilon_M$ and making $\varepsilon_M$ smaller if necessary we obtain
\begin{align*}
 M (\tau) \leq \frac{(1+C  \varepsilon_M) \varepsilon_M}{2-C\varepsilon_M} e^{(-1 + C \varepsilon)(\tau-\tau_0)}\leq\varepsilon_M e^{(-1 + C \varepsilon)(\tau-\tau_0)},
\end{align*}
which is an improvement of the bootstrap assumption \eqref{bootm} for any $\tau>\tau_0$ and which proves \eqref{estM} provided we improve the remaining bootstrap assumptions \eqref{bootsigma}--\eqref{bootY}.

\subsection{Suppressing the oscillations}\label{suppress}

Now we proceed to redefine ${\Sigma}_+$, $X$, $w_+$, $\Omega$ and $w_-$ to suppress the oscillations in these variables.
Let us make a similar change of variables as in \cite{NHW} (the change of variables for $\Sigma_+$, $X$, $M$ and $\Omega$ is the same as in \cite{NHW} if we set $w_-=0$):
\begin{align}
    &\bar{\Sigma}_+=\Sigma_+-\frac14 M(1+\Sigma_+)X^2\sin 2\psi,\\
    &\bar{X}=\frac{X}{1+\frac14M(X^2-1-\Sigma_+)\sin 2\psi}-\frac32Mw_-\Omega \sin \psi,\\
      \label{w+bar}&\bar{w}_+ = {w}_+-\frac12 M\left[w_- ( 6w_+- 4)+\beta \right]X\sin \psi,\\
      &\bar{\Omega}=\frac{\Omega}{1+\frac12 MX^2\sin 2\psi- 3 M X w_- \sin \psi},\\
    &\bar{w}_- = {w}_--\frac12 M\left(6w_-^2-\frac{8}{15}-2w_++\delta \right) X \sin \psi. 
\end{align}
Note that all the relevant quantities are bounded. Also the derivatives of $\alpha$, $\beta$, $\gamma$ and $\delta$ are bounded using \eqref{xixi}. Since we assume that $M$ is small we have that $\bar{\Sigma}_+$, $\bar{X}$, $\bar{M}$, $\bar{\Omega}$, $\bar{w}_+$, $\bar{w}_-$ are bounded.

Using  \eqref{M}--\eqref{w-} the evolution equations of all the barred quantities are as follows:

\begin{align}
&\bar{M}'=-\bar{M}\left[(1+\bar{\Sigma}_+)^2+O(M)\right],\\
\label{barsigma}&\bar{\Sigma}_+'=-\bar{X}^2-\bar{\Sigma}_+(1 -\bar{\Sigma}_+^2) +3\bar{w}_+\bar{\Omega}+ O(M),\\
\label{barX}&\bar{X}'=[\bar{\Sigma}_+(1+\bar{\Sigma}_+) +O(M)]\bar{X}+O(M)+O \left(\frac{Mw_-^2}{X} \right) ,\\
\label{barOmega}&\bar{\Omega}' =  2\left[ \bar{\Sigma}_+^2-3 \bar{\Sigma}_+ \bar{w}_+ +O(M)\right]\bar{\Omega},\\
&\label{barw+}\bar{w}_+'= \left(-\frac{8}{15}+\alpha+6\bar{w}_+^2\right)\bar{\Sigma}_+ +O(M),\\
&\label{barw-}\bar{w}_-'= \left[\bar{w}_-(-2+6\bar{w}_+)+\gamma \right]\bar{\Sigma}_+ +O(M).
\end{align}
We also want to suppress the oscillations for the terms concerning $\tilde{{\xi}}_{22}^{22}$, $\tilde{{\xi}}_{33}^{33}$, $\tilde{{\xi}}_{21}^{21}$, $\tilde{{\xi}}_{31}^{31}$, $\tilde{{\xi}}_{32}^{32}$. The equations \eqref{xi1}--\eqref{xi5} are
\begin{align*}
(\tilde{\xi}^{ij}_{ij})' = f_1(w_+,w_-,\tilde{\xi}_{ij}^{ij},\Phi^{ijb}_{ijb}) \Sigma_+ + f_2(w_+,w_-,\tilde{\xi}_{ij}^{ij},\Phi^{ijb}_{ijb} )X \cos \psi, 
\end{align*}
for some functions $f_1$ and $f_2$ of the given form (no summation of the indices). Thus, making the change of variable
\begin{align*}
{\bar{\xi}}^{ij}_{ij} =\tilde {\xi}^{ij}_{ij} - \frac12 M f_2 X \sin \psi,
\end{align*}
we obtain
\begin{align}\label{xibar}
({\bar{\xi}}^{ij}_{ij})' = f_1(\bar{w}_+,\bar{w}_-, {\bar{\xi}}^{ij}_{ij}, \Phi^{ijb}_{ijb}) \bar{\Sigma}_+ + O(M).
\end{align}

\subsection{The truncated system}\label{truncated}

As a first step we consider the truncated system formed by the evolution equations of $\bar{\Sigma}_+$ \eqref{barsigma}, $\bar{X}$ \eqref{barX}  and $\bar{w}_+$ \eqref{barw+} where we neglect the oscillatory terms with $M$ and the term $\alpha$ which is related to higher order momenta. Note that $\bar{w}_- $ basically appears only in the evolution equation of $\bar{w}_- $ and we ignore that equation for the moment.

To indicate that we are working with the truncated system we will use a hat on the variables and then the system is
\begin{align}
\label{tru1}&\hat{\Sigma}_+'=-\hat{X}^2-\hat{\Sigma}_+(1 -\hat{\Sigma}_+^2) +3\hat{w}_+(1-\hat{X}^2-\hat{\Sigma}_+^2),\\
\label{tru2}&\hat{X}'=[\hat{\Sigma}_+(1+\hat{\Sigma}_+)]\hat{X} ,\\
\label{tru3}&\hat{w}_+'= \left(-\frac{8}{15}+6\hat{w}_+^2\right)\hat{\Sigma}_+ .
\end{align}

Note that the curve $ -\frac13\hat{X}^2+\hat{w}_+(1-\hat{X}^2)=0$ at $\hat{\Sigma}_+=0$ is an invariant manifold since all derivatives of any order vanish there. This means that there is an equilibrium value for $\hat{w}_+$ which is
\begin{align*}
\hat{w}_+ =\frac{\hat{X}^2}{3(1-\hat{X}^2)}.
\end{align*}
This is the motivation for the introduction of the variable $Y$ \eqref{Y}, which in the truncated case we denote with a hat: 

\begin{align}\label{hatY}
\hat{Y}=-\frac13 \hat{X}^2+\hat{w}_+\left(1-\hat{X}^2\right).
\end{align}
We have the following lemma:

\begin{lem}\label{lemtru}
Consider the system of differential equations \eqref{tru1}--\eqref{tru3} and the function
\begin{align}\label{hatu}
\hat{u} = (2\hat{X}^2+8) \hat{\Sigma}_+^2 -15 \hat{\Sigma}_+ \hat{Y} + 45 \hat{Y}^2.
\end{align}
Then $\hat{u}'+\hat{u}=O(\hat{\Sigma}_+^3+\hat{\Sigma}_+^2\hat{Y}+ \hat{\Sigma}_+\hat{Y}^2)$.
\end{lem}

\begin{proof}
The proof is based on a straightforward computation.\\

The evolution equation for $\hat{\Sigma}_+$ using the variable $\hat{Y}$ is then
\begin{align*}
\hat{\Sigma}_+'=3 \hat{Y}-\hat{\Sigma}_+(1 -\hat{\Sigma}_+^2) -3\hat{w}_+\hat{\Sigma}_+^2. 
\end{align*}
From \eqref{hatY} we obtain the evolution equation for $\hat{Y}$ which is
\begin{align*}
\hat{Y}' & = \hat{\Sigma}_+ \left[ - 2\hat{X}^2 (1+\hat{\Sigma}_+)\left(\frac13+\hat{w}_+\right)+(1- \hat{X}^2)  \left(-\frac{8}{15}+6\hat{w}_+^2\right)\right] \\
& = \hat{\Sigma}_+ \left[ - 2\hat{X}^2 \hat{\Sigma}_+\left(\frac13+\hat{w}_+\right)-\frac{2}{15}\hat{X}^2-\frac{8}{15}+6\hat{Y}\hat{w}_+\right].
\end{align*}
Note that  $\hat{u}$ is zero for $\hat{\Sigma}_+=\hat{Y}=0$. In fact using Sylvester's criterion \cite{SR} since all leading principal minors of the matrix
\begin{align*}
\left(
\begin{matrix}
2\hat{X}^2+8&-\frac{15}{2}\\
-\frac{15}{2}&45
\end{matrix}
\right)
\end{align*}
are positive, $\hat{u}$ is positive definite.  The evolution equation for $\hat{u}$ is using \eqref{hatu} as follows:
\begin{align*}
\hat{u}' & = 4\hat{X}^2 (1+\hat{\Sigma}_+) \hat{\Sigma}_+^3 +  [2(2\hat{X}^2+8) \hat{\Sigma}_+-15\hat{Y}][ 3 \hat{Y}-\hat{\Sigma}_+(1 -\hat{\Sigma}_+^2) -3\hat{w}_+\hat{\Sigma}_+^2] \\
&\quad +(90\hat{Y} - 15 \hat{\Sigma}_+ )  \hat{\Sigma}_+ \left[ - 2\hat{X}^2 \hat{\Sigma}_+\left(\frac13+\hat{w}_+\right)-\frac{2}{15}\hat{X}^2-\frac{8}{15}+6\hat{Y}\hat{w}_+\right]. 
\end{align*}
Hence, we have
\begin{align*}
\hat{u}' +\hat{u} & = [4 \hat{X}^2(1+\hat{\Sigma}_+)   + 2 (2\hat{X}^2+8)   (\hat{\Sigma}_+ - 3 \hat{w}_+) -15\hat{Y}+10\hat{X}^2+30\hat{X}^2\hat{w}_+]\hat{\Sigma}_+^3  \\
& \quad - \left[45 \hat{w}_+ +180 \hat{X}^2 \left(\frac13 +\hat{w}_+ \right) \right]\hat{Y} \hat{\Sigma}_+^2  + 540\hat{w}_+  \hat{\Sigma}_+ \hat{Y}^2,
\end{align*}
which proves the lemma.
\end{proof}

Using the numerical solver for systems of ordinary differential equations of SciPy which is based on a Runge-Kutta-Fehlberg 
method (also called RK45 or explicit Runge-Kutta method of order 5(4)) \cite{DP}
one can solve easily the truncated system for some random numbers. Below, one can see that the variables converge rapidly to the equilibrium values.

\begin{center}
		{\includegraphics[scale=0.5]{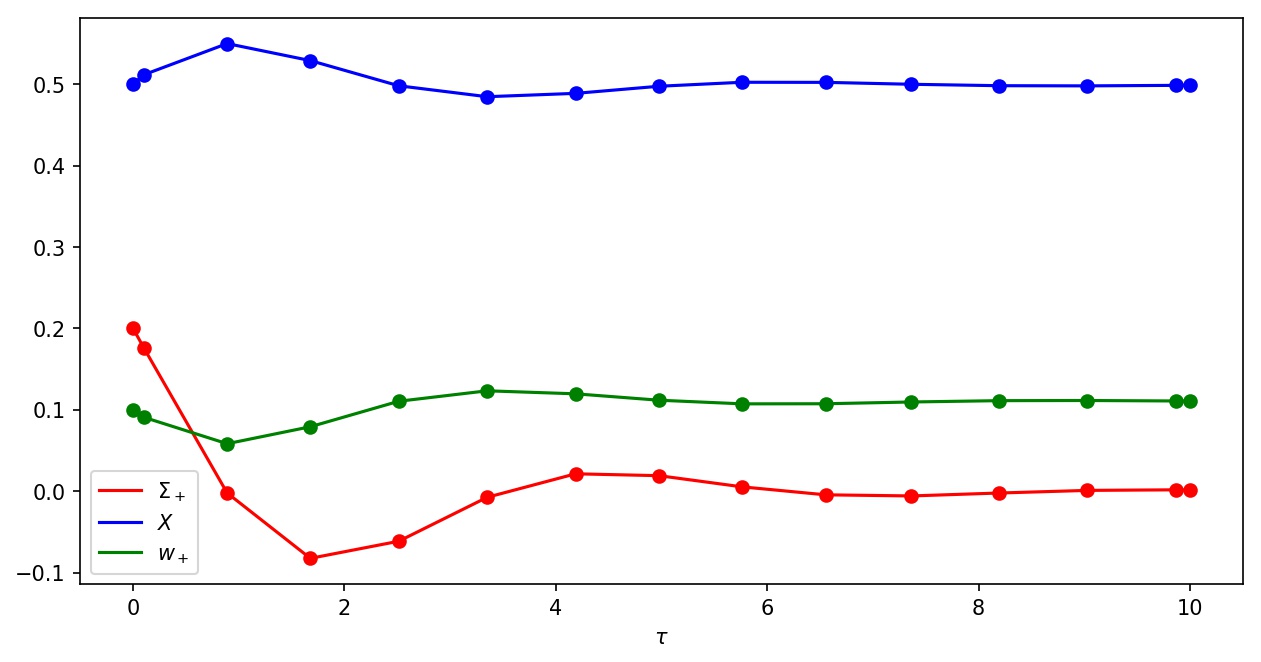}}
		\\[1ex]
		{Plot of the evolution of solutions to the truncated system.}
	\end{center}

\subsection{Closing the bootstrap argument} \label{closing}

In order to close the bootstrap argument we generalise Lemma \ref{lemtru} which corresponds to the truncated system to the general system using the estimate for $M$ and bounds for the matter terms.

Note that all the smallness assumptions for the different variables imply that the corresponding barred quantities are also small making $M$ smaller if necessary.

\subsubsection{Bounds for the matter terms}

For the matter terms it will be sufficient to show that they can be bounded by small quantities. From \eqref{xibar} and since $w_+$, $w_-$, $\xi^{ij}_{ij}$, $\Phi^{ij}_{ij}$ are bounded, we have using \eqref{estM} and \eqref{bootsigma} 
\begin{align*}
(\bar{\xi}^{ij}_{ij})' \leq C (\vert \bar{\Sigma}_+ \vert +  M) \leq C [\varepsilon_{\Sigma} e^{(-\frac38+ \varepsilon)\tau}+ \varepsilon e^{(-1+\varepsilon) \tau}],
\end{align*}
after integration yields
\begin{align*}
\bar{\xi}^{ij}_{ij}(\tau_0)- C \varepsilon \leq \bar{\xi}^{ij}_{ij}(\tau) \leq \bar{\xi}^{ij}_{ij}(\tau_0)+ C \varepsilon,
\end{align*}
and also
\begin{align}\label{boundxi}
\tilde{\xi}^{ij}_{ij} (\tau_0) - C \varepsilon \leq \tilde{\xi}^{ij}_{ij}(\tau) \leq \tilde{\xi}^{ij}_{ij} (\tau_0) +C \varepsilon.
\end{align}
As a consequence $\alpha$ and $\gamma$, defined in \eqref{alpha} and \eqref{gamma}, have a similar bound.

From \eqref{barw-} we obtain the same type of bound having in mind that $\gamma$ is bounded, obtaining:
\begin{align}\label{boundw-}
\bar{w}_-(\tau_0) - C \varepsilon \leq \bar{w}_- (\tau) \leq \bar{w}_-(\tau_0) +C \varepsilon.
\end{align}
This estimate will be particularly important in the following. It implies that $\bar{w}_-$ remains close to their initial value by choosing $ \varepsilon $ small. In particular, we may use the triangle inequality to obtain
\[
| { \bar w }_- ( \tau ) | \geq | { \bar w }_- ( \tau_0 ) | - | { \bar w }_- ( \tau_0 ) - { \bar w }_- ( \tau ) | \geq | { \bar w }_- ( \tau_0 ) | - C \varepsilon \geq \frac12 | { \bar w }_- ( \tau_0 ) |, 
\]
where the last inequality holds for sufficiently small $ \varepsilon > 0 $. In a similar way, we obtain
\begin{align}\label{boundabsw-}
\frac12 | { \bar w }_- ( \tau_0 ) | \leq | { \bar w }_- ( \tau_0 ) | \leq 2 | { \bar w }_- ( \tau_0 ) |. 
\end{align}

\subsubsection{Estimate of the quotient of $w_-$ and $X$}

Now, consider the evolution equation of $B=\frac{\bar{w}_-}{\bar{X}}$,  using the evolution equations for $\bar{X}$ \eqref{barX} and $\bar{w}_-$ \eqref{barw-} :
\begin{align*}
B' = B\left (-3+6\bar{w}_+-\bar{\Sigma}_++\frac{\gamma }{\bar{w}_-}\right) \bar{\Sigma}_+ +O \left(\frac{BM}{\bar{w}_-}\right)+O(BM)
+O\left(\frac{B^2M}{\bar{w}_-}\right)+O (M B^3),
\end{align*}
from which it follows that
\begin{align}\label{estB}
\frac{B'}{B}= \left(-3+6\bar{w}_+-\bar{\Sigma}_++\frac{\gamma }{\bar{w}_-}\right)\bar{\Sigma}_+ +O \left(\frac{M}{\bar{w}_-}\right)+O(M)
+O\left(\frac{BM}{\bar{w}_-}\right)+O (M B^2).
\end{align}
We show now that the right hand side of \eqref{estB} is a small quantity, smaller than some $\varepsilon$ using the bootstrap assumptions for $\Sigma_+$ and $\frac{w_-}{X}$ \eqref{bootsigma}, \eqref{bootw-X}, the estimates obtained for $\gamma$ and $w_-$ \eqref{boundxi}, \eqref{boundw-} and the obtained estimate of $M$ \eqref{estM} and since
\begin{align}\label{ineq1}
\left\vert \frac{\gamma \bar{\Sigma}_+ }{\bar{w}_-} \right\vert &\leq \frac{ C \vert \gamma(\tau_0) \vert   \varepsilon e^{(-\frac38+\varepsilon)\tau}}{\vert \bar{w}_-(\tau_0) \vert},\\
\frac{M}{\vert \bar{w}_- \vert}& \leq \frac{\varepsilon e^{(-1+\varepsilon)\tau}}{\vert \bar{w}_-(\tau_0)\vert },\\
\frac{\vert B \vert M}{\vert \bar{w}_- \vert } & \leq  \frac{ C\varepsilon e^{(-\frac34+\varepsilon)\tau}}{\vert \bar{w}_-(\tau_0)\vert},\\
\label{ineq4}MB^2 & \leq C \varepsilon e^{(-\frac12+\varepsilon)\tau}.
\end{align}
We have assumed that $w(\tau_0)$ is small, which by choosing $\varepsilon$ sufficiently small, implies that $\bar{w}_-(\tau_0) \neq 0$. The inverse of  $\bar{w}_-(\tau_0)$ might not be a small quantity, but since it appears always multiplied with $\varepsilon$ the terms on the right hand side of the inequalities \eqref{ineq1}--\eqref{ineq4} are all small. In fact they are integrable, so that integrating \eqref{estB} we obtain
\begin{align}\label{BB}
\vert B \vert = \left \vert \frac{\bar{w}_-}{\bar{X}} \right \vert \leq \left \vert \frac{\bar{w}_-}{\bar{X}}(\tau_0) \right \vert C .
\end{align}
Choosing $\varepsilon$ small if necessary this improves the bootstrap assumption for $\frac{w_-}{X}$ giving
\begin{align}\label{estw-X}
\frac{w_-  }{X}(\tau) = O(1).
\end{align}

\subsubsection{Estimate of $Y$}

For the estimate of $Y$ we work with the corresponding barred quantity:
\begin{align*}
 \bar{Y}= -\frac13{\bar{X}}^2+\bar{w}_+(1-\bar{X}^2)
\end{align*}
and we define:
\begin{align*}
u= (2\bar{X}^2+8) \bar{\Sigma}_+^2 -15 \bar{\Sigma}_+ \bar{Y} + 45 {\bar{Y}}^2.
\end{align*}
We want to estimate the derivative of $u$. Since we can use Lemma \ref{lemtru}, we have to focus only on the terms which are not present in the truncated case. The result after some computations is that 
\begin{align*}
u' \leq (-1+C\varepsilon) u+  O \left(\bar{\Sigma}_+^3+\bar{\Sigma}_+^2\bar{Y}+ \bar{\Sigma}_+\bar{Y}^2 + \left(  \bar{\Sigma}_+ +  \bar{\Sigma}_+ B+  \bar{Y}+\bar{Y}B \right)M\right),
\end{align*}
where we have used the estimates \eqref{boundxi} and \eqref{estw-X} and the fact that we can choose the initial data for $\vert \tilde{\xi}_{22}^{22}(\tau_0) \vert $,  $\vert \tilde{\xi}_{33}^{33}(\tau_0)\vert $,  $\vert \tilde{\xi}_{21}^{21}(\tau_0)\vert $, $\vert \tilde{\xi}_{31}^{31}(\tau_0) \vert$, $\vert \tilde{\xi}_{32}^{32}(\tau_0) \vert$ to be smaller than $\varepsilon$, so that $\alpha \leq C\varepsilon$.

Now using the bootstrap assumptions and the estimate \eqref{estw-X} we obtain that that $u= O(\varepsilon^2 e^{(-1+C\varepsilon)\tau})$ and 
\begin{align*}
& \bar{\Sigma}_+ = O(\varepsilon e^{(-\frac12+C\varepsilon)\tau}),\\
& \bar{Y} = O(\varepsilon e^{(-\frac12+C\varepsilon)\tau}).
\end{align*}
Choosing $\varepsilon$ small if necessary gives the estimates for $\Sigma_+$ and $Y$ which closes that bootstrap argument and finishes the proof of the main theorem.

\subsection{The Weyl parameter and the shear at late times}\label{weyl}

The Weyl parameter $\mathcal{W}$ or Hubble normalised Weyl curvature has the following expression (cf. (3.39) of \cite{WHU} and references therein for details):
\begin{align}\label{par}
\mathcal{W}= \frac{2X}{M} [1+O(M)].
\end{align}

We conclude with the following corollary:
\begin{cor*}
Consider the same assumptions as in Theorem \ref{bootthm}. Then, we have
\begin{align}\label{estW}
\mathcal{W} \geq C e^{(1-\varepsilon)\tau},
\end{align}
and
\begin{align*}
&\lim_{\tau \rightarrow \infty} M = \lim_{\tau \rightarrow \infty } \Sigma_+ = 0,\\
&\lim_{\tau \rightarrow \infty} X = X_{\infty}, \quad \lim_{\tau \rightarrow \infty} w_+ = (w_+)_\infty,\\
&\lim_{\tau \rightarrow \infty} w_- = (w_-)_{\infty},
\end{align*}
with
\begin{align}\label{rel}
(w_+)_{\infty}= \frac{X_{\infty}^2}{3(1-X_{\infty}^2)}.
\end{align}

\end{cor*}

\begin{proof}
Using the estimate \eqref{BB} and the estimate for $\bar{w}_-$ \eqref{boundw-} making $\varepsilon$ small if necessary, we obtain a lower bound for $X$
\begin{align}\label{lowerX}
X \geq  C X(\tau_0) \left\vert \frac{ w_-(\tau)}{w_-(\tau_0)}\right \vert \geq C X(\tau_0) \left \vert 1+ O\left(\frac{\varepsilon}{w_-(\tau_0)}\right) \right \vert .
\end{align}
This inequality together with the estimate for $M$ \eqref{estM} give the desired conclusion for the Hubble normalised Weyl curvature \eqref{estW}. 

Note that in fact $X$ will be small for all times. From \eqref{Y} we can express $X$ in terms of $Y$ and $w_+$:
\begin{align*}
{X}^2 = \frac{w_+-Y}{\frac13+w_+}.
\end{align*}
Since we have shown that $w_+$ will be small for all times and $Y$ will tend to zero, we have that $X$ will be the square of a small quantity. In particular, we have that $X\simeq \sqrt{3w_+}$.
 Moreover due to \eqref{estw-X} $w_-$ will be bounded by the same order as $X$.

In order to prove the second part of the corollary, note that the second derivatives of the barred quantities $\bar{X}$, $\bar{\Sigma}_+$, $\bar{M}$, $\bar{w}_+$ and $\bar{w}_-$ are bounded. The terms we have collected in $O(M)$ in the differential equations for the barred quantities \eqref{barM}, \eqref{barX}, \eqref{barsigma}, \eqref{barw+}, \eqref{barw-} are polynomials of $\bar{X}$, $\bar{\Sigma}_+$, $\bar{M}$, $\bar{w}_+$, $\bar{w}_-$, $\cos \psi$ and $\sin \psi$, all multiplied by $M$. All the possible derivatives of these quantities are bounded, since the only problematic term comes from the derivative of $\psi$. The quantity $\frac{1}{M}$ is absorbed by the factor $M$ and the quantity $\frac{1}{X}$ is now not problematic anymore since we have shown that it $X$ has the lower bound \eqref{lowerX}. Since $\bar{X}$, $\bar{\Sigma}_+$, $\bar{M}$, $\bar{w}_+$, $\bar{w}_-$ and their first and second derivatives are bounded we can apply the Arzel\`a-Ascoli theorem.
The limits for $\bar{M}$, $\bar{\Sigma}_+$ are obtained from the estimates \eqref{estM}, \eqref{estsigma} of the theorem. Taking the limit in the differential equations for $\bar{X}$, $\bar{w}_+$ and $\bar{w}_-$ we obtain that these quantities also have limits. Taking the limits in the definitions of $\bar{X}$, $\bar{\Sigma}_+$, $\bar{M}$, $\bar{w}_+$, $\bar{w}_-$ and $\bar{Y}$ in terms of the unbarred quantities together with the limit of $M$ we obtain the desired conclusion.
\end{proof}

\section{Conclusion and Discussion}

Our results show that one can have a homogeneous cosmological model with massless particles where the shear is arbitrarily small and remains small, and the anisotropy of the matter distribution is arbitrarily small and remains small, but where the space-time is far from and remains far from an isotropic spacetime as the corollary shows.

In particular we have {described} the asymptotic behaviour of reflection symmetric solutions to the Einstein-Vlasov system with Bianchi VII$_0$ symmetry. The solutions tend to the equilibrium solution of what we have called the truncated system assuming small data.

In the massive case it was shown that reflection symmetric solutions to the Einstein-Vlasov system with Bianchi VII$_0$ symmetry behave as Einstein-dust solutions assuming small data (cf. Theorem 1 of \cite{LN2}). 

In the massless case, the relation between pressure (trace of energy momentum tensor) and energy density is fixed and given by \eqref{massless}. However here we have shown that in the massless case, the behaviour differs from that of a radiation fluid, in the sense that we obtain an exponential decay rate, while the decay rate for a radiation fluid, both tilted and non-tilted, is polynomial \cite{LDW,NHW,WHU}. 

Another difference is that in the collisionless case the shear will always remain small, but will not tend to zero. If we come back to the question of what happens with the spacetime if the matter distribution is almost isotropic, the answer to this question would be in the case considered here given by our corollary. In particular the anisotropy of the Hubble normalised trace free part of the second fundamental form is related to the anisotropy of the energy momentum tensor by \eqref{rel}.

To summarise, there are important quantitive and qualitative consequences when analysing the future of solutions to the Einstein equations with Bianchi VII$_0$ symmetry coupled to a perfect fluid or to collisionless matter. It would be very interesting to analyse the Einstein equations with this symmetry coupled to the Boltzmann equation. The different collision kernels might make important differences. However if collisions between the particles are allowed, the reflection symmetry of the particle distribution function assumed here, will not be preserved. Thus an interesting future project would be to first remove this symmetry assumption and to consider the non-diagonal case as was done previously for Bianchi II and VI$_0$ spacetimes \cite{E4} and second, generalise to the Boltzmann case as was done for Bianchi I spacetimes \cite{LN}.
 
\section*{Acknowledgements}

The authors thank Paul Tod from the University of Oxford for helpful discussions. H. Lee was supported by the Basic Science Research Program through the
National Research Foundation of Korea (NRF) funded by the Ministry of Science, ICT \& Future Planning (NRF-2018R1A1A1A05078275). {E. Nungesser has been supported by Grant MTM2017-85934-C3-3-P of Agencia Estatal de Investigaci\'{o}n (Spain).


\begin{thebibliography}{10}

\bibitem{B}
H.~Barzegar.
\newblock{Future attractors of Bianchi types II and V cosmologies with massless Vlasov matter.}
\newblock{\em Class. Quant. Grav.} 38 065019, 2021.

\bibitem{BFH}
H.~Barzegar, D.~ Fajman, G.~Hei\ss el.
\newblock{Isotropization of Slowly Expanding Spacetimes.}
\newblock{\em Phys. Rev. D} 101, 044046, 2020.


\bibitem{CH}
S.~Calogero and J.M.~Heinzle.
\newblock{Bianchi Cosmologies with Anisotropic Matter: Locally
                  Rotationally Symmetric Models.}
\newblock {\em Physica} D240, 636--660, 2011. 


\bibitem{CHe}
A.~Coley and S.~Hervik.
\newblock{Dynamical systems approach to the tilted Bianchi models of solvable type.}
\newblock{\em Class. Quant. Grav.} 22, 579--606, 2005.

\bibitem{DP}

J.R.~Dormand and P. J.~Prince.
 \newblock{A family of embedded Runge-Kutta formulae.} 
 \newblock{\em J. Comput. Appl. Math.} 6,1, 9-26, 1980.

\bibitem{EGS}

J.~Ehlers, P.~Geren, and R. K.~Sachs. 
\newblock{Isotropic solutions of the Einstein-Liouville Equations.}
\newblock{\em J. Math. Phys.}, 9,9: 1344--1349, 1968.

\bibitem{HervikVII0}
S.~Hervik, R.J.~ van den Hoogen, W.C.Lim and A.A.~Coley.
 \newblock{The Futures of Bianchi type VII$_0$ cosmologies with vorticity.}
 \newblock{\em Class. Quant. Grav.} 23, 845--866, 2006.
 
\bibitem{JTV}
J.~Joudioux, M.~Thaller, and Juan A.~Valiente Kroon.
\newblock{The Conformal Einstein Field Equations with Massless Vlasov Matter.}
\newblock{\em Ann. Inst. Fourier} 71, 2, 799--842, 2021.

\bibitem{LN}
H.~Lee and E.~Nungesser.
\newblock {Future global existence and asymptotic behaviour of solutions to the Einstein-Boltzmann system with Bianchi I symmetry.}
\newblock{\em J. Differ. Equations} 262, 11: 5425--5467, 2017.

\bibitem{LN2}
H.~Lee and E.~Nungesser.
\newblock {Self-similarity breaking of cosmological solutions with collisionless matter.}
\newblock {\em Ann. Henri Poincare} 19, 7:2137–-2155, 2018. 

\bibitem{LNT2}
H.~Lee and E.~Nungesser and K.P.~Tod.
\newblock {On the future of solutions to the massless Einstein-Vlasov system in a Bianchi I cosmology.}
\newblock {\em Gen. Rel. Grav.} 52, 48,  2020.

\bibitem{LDW}
W.C.~Lim, R.J.~Deeley and J.~Wainwright.
\newblock {Tilted Bianchi VII$_0$ cosmologies-the radiation bifurcation.}
\newblock {\em Class. Quant. Grav.}, 23:3215--3234, 2006.

\bibitem{NHW}
U.~Nilsson, M.~Hancock, and J.~Wainwright.
\newblock {Non-tilted Bianchi VII$_0$ models - the radiation fluid.}
\newblock {\em Class. Quant. Grav.}, 17:3119--3134, 2000.

\bibitem{NUWL}
U.~Nilsson, C.~Uggla, J.~Wainwright, W.C.~Lim.
\newblock{An almost isotropic cosmic microwave temperature does not imply an almost isotropic universe.}
\newblock{\em Astrophys. J.} 522, L1, 1999.

\bibitem{E4}
E.~Nungesser.
\newblock{Future non-linear stability for solutions of the Einstein-Vlasov system of Bianchi types II and VI$_0$.}
\newblock{\em J. Math. Phys.} 53, 102503, 2012.
     
 \bibitem{R}
 S.~R\"{a}s\"{a}nen.
 \newblock{On the relation between the isotropy of the CMB and the geometry of the universe.}
 \newblock{\em Phys. Rev. D}, 79: 123522,  2009.
     
\bibitem{Ring}
H.~Ringstr{\"{o}}m.
\newblock {\em On the Topology and Future Stability of the Universe.}
\newblock Oxford University Press, Oxford, 2013.

\bibitem{Rendall}
A.D.~Rendall. 
\newblock{The Initial singularity in solutions of the Einstein- Vlasov system of Bianchi type I.}
 \newblock{\em J. Math. Phys.}, 37:438--451, 1996.

\bibitem{SR}
I.R.~Shafarevich and A.O.~Remizov.
\newblock{Linear Algebra and Geometry.}
\newblock{Springer Verlag, Berlin Heidelberg}, 2013.

\bibitem{SME}
W.~Stoeger, R.~Maartens, G.F.R.~Ellis.
\newblock{Proving Almost-Homogeneity of the Universe: an Almost Ehlers-Geren-Sachs Theorem.}
\newblock{\em Astrophys. J.} 443, 1, 1995.

\bibitem{TE}
R. Treciokas and G.F.R. ~Ellis. 
\newblock{Isotropic solutions of the Einstein-Boltzmann Equations.}
\newblock{\em Comm. Math. Phys.}, 23:1–22, 1971.


\bibitem{WHU}
J.~Wainwright, M.J.~Hancock, and C.~Uggla.
\newblock {Asymptotic self-similarity breaking at late times in cosmology.}
\newblock {\em Class. Quant. Grav.}, 16:2577--2598, 1999.




\end{thebibliography}
\end{document}